\documentclass[preprint,12pt]{elsarticle}
\usepackage{graphicx}      
\usepackage{soul}
\usepackage{latexsym}
\usepackage{amsmath}
\usepackage{amsfonts}
\usepackage{epstopdf}
\usepackage{amssymb}
\usepackage[english]{babel}

\usepackage{amsthm}
\usepackage{graphicx,subfig}
 
\newtheorem{theorem}{Theorem}
\newtheorem{lemma}[theorem]{Lemma}
\newtheorem{proposition}[theorem]{Proposition}

\newtheorem{remark}[theorem]{Remark }
\newtheorem{definition}[theorem]{Definition}
\newtheorem{example}[theorem]{Example}

\usepackage{color}
 
\usepackage{xcolor}
\definecolor{verde}{rgb}{0,0.7,0}

\usepackage[colorlinks=true,linkcolor=blue,urlcolor=blue,citecolor=blue]{hyperref}

\newcommand{\be}{\begin{equation}}
\newcommand{\ee}{\end{equation}}

\newcommand{\bea}{\begin{eqnarray}}
\newcommand{\eea}{\end{eqnarray}}
\newcommand{\bean}{\begin{eqnarray*}}
\newcommand{\eean}{\end{eqnarray*}}

\journal{European Journal of Control} 
\begin{document}

\begin{frontmatter}

\title{A Bandwagon Bias Based Model for Opinion Dynamics: \\
 Intertwining between Homophily and Influence Mechanisms }

 \author{Giulia De Pasquale and Maria Elena Valcher\\
Dept. of Information Engineering, University of Padova\\
via Gradenigo 6B, 35131 Padova, Italy, \texttt{giulia.depasquale@phd.unipd.it, meme@dei.unipd.it}}

\begin{abstract}                          
Recently a model for the interplay between homophily-based appraisal dynamics and influence-based opinion dynamics has been proposed.
The model explores for the first time how the opinions of a group of agents on a certain number of issues/topics is influenced by the agents' mutual appraisal and, conversely, the agents' mutual appraisal is updated based on the agents' opinions on the various issues, according to a homophily model.
In this paper we show that a simplified (and, in some situations, more feasible) 
version of the model, that accounts only for the signs of the agents' appraisals rather than for their numerical values, provides an equally accurate and effective model of the opinion dynamics  in small networks. The  equilibria reached by this model 
 correspond,  almost surely,  to situations in which the agents' network is complete and structurally balanced. On the other hand, 
we ensure that such equlibria can always be reached in a finite number of steps, and,  differently from the original model, we rule out other types of  equilibria that correspond to disconnected social networks. 
\end{abstract}

\begin{keyword}
Opinion dynamics, homophily model, equilibria, structural balance.
\end{keyword}
\end{frontmatter}
\section{Introduction} 
 
Over the last few decades, the modelling and analysis of sociological phenomena have attracted the interests of researchers from various fields, such as sociology, economics, and mathematics  \cite{Altafini2013,friedkin,sociological_model,hiller}. 
In several cases,   sociological models represent the primary focus of the investigation, but there are numerous contexts, such as product promotion, spread of diseases, resource allocation, etc.,  where  social dynamics   represents the context  in  which other
phenomena evolve.  Consequently, understanding its   behaviour 
 represents a preliminary but   fundamental step   in order to investigate and understand
 the evolution of the process of interest  \cite{Aghbolagh19,Marvel11}.
 As a result, it becomes of great importance to build a reliable model for the social dynamics, that allows to forecast the network evolution and thus  to design strategies aimed at driving the network towards the desired  configuration  \cite{ProskurnikovTempo}. 
 Dynamic social balance theory    is concerned with the study and analysis of the evolution of socially unbalanced networks towards socially balanced ones, namely   networks in  balanced  configurations in which all the agents split in (at most) two groups in   such a way   that all the agents in the same group have friendly relationships, while agents from different groups    have not \cite{Harary,Heider}. 
 
Even if, from a    modeling perspective, 
 the study of social balance has rather remote origins, as  witnessed by the pioneering works of Heider \cite{Heider}, Cartwright and Harary \cite{FH:59,Harary}, DeGroot \cite{Degroot}, the dynamic social balance theory   represents an active and timely research topic. In this regard we mention the recent works of Mei et al. \cite{HomophilyMei} in which two dynamical models based on two different social mechanisms, the homophily mechanism and the influence mechanism, are proposed. In  the homophily mechanism,  individuals update their mutual appraisals based on their appraisals   of the other  group members. In the influence mechanism, instead, each agent attributes an influence to the other network members, based on the appraisal that the agent has about them.  Reference \cite{HomophilyMei}  shows that both   mechanisms drive the network towards   social balance,   but the homophily mechanism gives a more general explanation for the emergence of the social balance with respect to the influence one. 
 \\
   A relevant contribution to the dynamic social balance analysis is the one from Quattrocchi et al. \cite{quattrocchi} whose model takes into account the presence of media and gossip as separate mechanisms. 
   Another inspiring work in which a sociological mathematical model, including two coexisting social mechanisms, is studied, is the recent work from Liu et al. \cite{MeiDorfler}. 
 
 In \cite{MeiDorfler} a novel model in which the interpersonal appraisals and the individual opinions evolve according to an intertwined dynamics has been proposed for the first time. In the proposed state space model, the   authors assume as state variables both the interpersonal appraisals   and the agent's opinions, namely  the opinions of   the agents on a specific selection of topics. Specifically, the   model  relies on the assumption that the opinion that an agent has on a particular issue is the   (signed) weighted average of the opinions that all the other agents have on that issue,   where the weights are the appraisals that the agent has about each of them. At the same time, the interpersonal relationship of an agent pair depends on the comparison between the opinions that the two agents have about all the topics into play, thus following a homophily mechanism.
 This model  
can be interpreted as a mathematical formalization  of a form of cognitive bias known in psychology as ``bandwagon bias" \cite{bandwagon}, by this meaning 
that our opinions on topics and issues are influenced by the opinions that other individuals have on the same topics and by the relationships we have
with those individuals.  
Bandwagon bias results in an intertwined dynamics involving both a homophily mechanisms for the interpersonal relationships and an influence mechanism for the agents' opinions. 
  This is in line with the fact that, in real life, interpersonal appraisals influence individual opinions and viceversa. 
Another dynamical model that studies how cognitive bias drives the formation of social influences can be found in \cite{BulloML}.


Inspired by   \cite{MeiDorfler} we propose here a  mixed-binary  and real-valued version of the aforementioned model, 
  by this meaning that while we assume that the agent's opinions take (positive or negative) real values, that represent their  levels of  appreciation or 
dislike of each specific issue, we do not quantify the level of mutual appraisal, but only take into account whether the mutual relationships between pairs of agents are friendly or hostile.
 This  simplified assumption has been   adopted in our previous work \cite{bin_hom}, as well as in  the works of Cisneros-Velarde et al. \cite{Bullo2020} and of Mei et al. \cite{not_all_to_all}   in which      mutual appraisals are treated as binary variables. 
  In particular, in \cite{not_all_to_all} the evolution of a signed unweighted non-all-to-all network towards the straightforward generalization of the concept of structural balance has been proposed, thus leading to the graph-theoretical concept of  ``triad-wise structural balance", where each agent's ego-networks satisfies the structural balance property. 
  This discrete-time ``gossip-like" model enjoys the property of convergence towards a non-all-to-all structural balance configuration, while    the structure of the associated   graph is time invariant, since only the signs of the weights change.   In \cite{Bullo2020} a network formation game, in which pairs of rational individuals strategically change the signs of the edges in a complete network is proposed. The game is shown to strategically reduce the cognitive dissonance in the network along time, by driving the network towards clustering balance \cite{Davis}. 
The motivation  behind the study of the dynamical evolution of unweighted signed social networks, that units the aforementioned works, comes from the fact that there are many circumstances in which recognizing the type  (friendly or hostile) of relationship between individuals is    easy, while  assessing its intensity 
is      complicated and prone to model errors.
In fact, while  individual evaluations of  certain products or their opinions on  certain topics can be easily obtained,  attributing  numerical values to the 
mutual appraisals   is  more challenging and oftentimes individuals prefer to not even reveal them.  \\
 We show that our simpler model retains all the good properties of the model proposed in \cite{MeiDorfler}  both in terms of transient behaviour and convergence to structurally balanced equilibria, meanwhile strengthening some of the results   derived for that model. In particular, 
  our model exhibit only two types of   long term behavior: 
either   the social network converges in finite time towards a socially balanced   all-to-all equilibrium  or asymptotically converges to zero. Other  equilibrium  structures, that arise for the model in \cite{MeiDorfler} and that
correspond to the case when the group of agents splits into disconnected structurally balanced subnetworks, are ruled out by our model assumptions,  which are designed for small networks.   In such contexts, getting an all-to-all equilibrium network is    realistic and, as it will be clear from simulations,   the situation when a structurally balanced equilibrium cannot be found and all individuals eventually weaken their opinions and appraisal to avoid long term conflicts (see \cite{Altafini2013}) is a very rare occurrence.
 
The paper is organized as follows: in Section \ref{model} the model is introduced and its equilibrium conditions are studied, Section \ref{convergence_FT} deals with the finite time behaviour of the model, while in Section \ref{convergence_inf} its asymptotic convergence properties are studied. Section \ref{conclusion} concludes the paper.
\medskip

{\bf Notation}.\
Given   $k, n\in \mathbb{Z}$, with $k <n$,   the symbol   $[k,n]$   denotes the  integer set  $\{k, k+1, \dots, n\}$.
We let   ${\bf e}_i$ denote the $i$-th vector of the canonical basis of $\mathbb{R}^n$,  where the dimension $n$ will be   clear from the context. 
  The vectors ${\bf 1}_n$ and ${\bf 0}_n$ denote the $n$-dimensional vectors whose entries are all $1$ or $0$, respectively.

The function ${\rm sgn}\!\!:\mathbb{R}^{n \times m}\rightarrow \{-1,0,1\}^{n \times m}$ is the   function that maps a real matrix into a matrix taking values in $\{-1,0,1\},$ in accordance with the sign of its entries.

In the sequel, the $(i,j)$-th entry of a matrix  ${\bf X}$ is denoted  either by ${ X}_{ij}$
or  by $[{\bf X}]_{ij}$, 
while the $i$-th entry of a vector ${\bf v}$ either by $v_i$ or by $[{\bf v}]_i$. 
 The notation  ${\bf X}= {\rm diag} \{x_1, x_2,  \dots,  x_N\}$ indicates the diagonal matrix whose diagonal entries are $x_1, x_2, \dots, x_N$. 
Given a matrix ${\bf X}\in {\mathbb R}^{N\times N}$, the {\em spectrum} of ${\bf X}$, $\sigma({\bf X})$, is the set of eigenvalues of ${\bf X}$.\\
An {\em undirected and signed  graph} is a triple~\cite{Mohar} $\mathcal{G}=(\mathcal{V},\mathcal{E},{\mathcal A})$, where $\mathcal{V}=\{1,\dots,N\}=[1,N]$ is the set of vertices, $\mathcal{E}\subseteq\mathcal{V}\times\mathcal{V}$  the set of arcs (edges), and
${\mathcal A}\in\{-1,0,1\}^{N\times N}$  the  {\em adjacency matrix} of the  graph $\mathcal{G}$. 
An arc $(j,i)$  belongs to ${\mathcal E}$ if and only if ${\mathcal A}_{ij} 
\ne 0$ and when so it may have either weight $1$ or weight $-1$.
As the graph is undirected,
 $(i,j)$ belongs to ${\mathcal E}$ if and only if $(j,i)\in {\mathcal E}$, and   
   they have the same weight (equivalently ${\mathcal A}$ is a symmetric matrix).
 A sequence 
 ${j_1}
\leftrightarrow {j_2} \leftrightarrow  {j_3}  \leftrightarrow \dots  \leftrightarrow {j_{k}} \leftrightarrow {j_{k+1}}$
is a {\em path}  
 of length $k$ 
 connecting ${j_1}$ and ${j_{k+1}}$
provided that
$({j_1},{j_2}), ({j_2},{j_3}),\dots,$ $({j_{k}}, {j_{k+1}}) \in {\mathcal E}$. A closed path  in which each node, except the start-end node, is distinct is called {\em cycle}, and a cycle of unitary length is also known as {\em self-loop}.
Since the adjacency matrix uniquely identifies the graph, in the following we will oftentimes use the notation   ${\mathcal G}({\mathcal A})$ to denote the graph having ${\mathcal A}\in\{-1,0,1\}^{N\times N}$ as adjacency matrix. \\
 The graph $\mathcal{G}$ is said to be {\em complete} if, for every pair of vertices $(i,j)$,  $i,j \in \mathcal{V}$, there is an edge connecting them, namely $(i,j) \in \mathcal{E}$.   If so, ${\mathcal E}= {\mathcal V}\times {\mathcal V}$ and ${\mathcal A}\in~\{-1,1\}^{N\times N}$. 
Given three distinct vertices $i,j$ and $k\in {\mathcal V}$,    
 the {\em triad} $(i,j,k)$  is said to be  {\em balanced} \cite{Bullo2020} if 
${\mathcal A}_{ij} {\mathcal A}_{jk} {\mathcal A}_{ki}=1$ and {\em unbalanced} if ${\mathcal A}_{ij} {\mathcal A}_{jk} {\mathcal A}_{ki}=-1$.

  In this work we    consider undirected 
  and signed graphs with unitary self loops.   Therefore the adjacency matrix of the graph  belongs to the set \cite{bin_hom}
\be
\mathcal{S}_{1}^N := \{ {\mathbf M} \in \{-1,0,1\}^{N \times N} : {\mathbf M}= {\mathbf M}^\top, \  M_{ii} =1, \  \forall i\in  [1,N] \}.
\ee
A graph ${\mathcal G}$ is said to be \textit{structurally balanced} if it   can be partitioned into two factions of vertices such that   edges between vertices of the same faction have nonnegative weights, and edges between vertices from different factions have nonpositive weights (see, also, the aforementioned concept of ``triad-wise structural balance" in \cite{not_all_to_all}). 
\\
The following result easily follows from Lemma 2.2 in  \cite{HomophilyMei}.

\begin{lemma}[Structural balance for complete graphs]\label{StructBal}
Given a matrix   ${\bf X}\in {\mathcal S}_1^N \cap \{-1,1\}^{N\times N}$, 
the following facts are equivalent: 
\begin{itemize}
\item[i)]  ${\bf X}={\bf p}{\bf p}^\top,$ for some vector ${\bf p}\in \{-1,1\}^N$;
\item[ii)] rank$({\bf X})=1$;
\item[ii)] for every $a,b\in [1,N]$ either  ${\bf e}_a^\top {\bf X} =  {\bf e}_b^\top {\bf X}$ or
${\bf e}_a^\top {\bf X} =  - {\bf e}_b^\top {\bf X}$ ;
\item[iv)] the graph ${\mathcal G}({\bf X})$ is structurally balanced;
\item[v)] all the triads $(i,j,k)$ of distinct vertices in ${\mathcal G}({\bf X})$ are balanced.
\end{itemize}
\end{lemma}

 In the following we will say that ${\mathbf X}$ is structurally balanced if ${\mathcal G}({\bf X})$ is structurally balanced.

\section{The model: properties, equilibrium points and periodic solutions} \label{model}

  Given a group  of $N$ agents, we denote by  ${\mathbf X}(t) \in \{-1,0,1\}^{N \times N}$   the   \textit{appraisal matrix at   time  $t$} of the agents, whose $(i,j)$-th entry represents  agent $i$'s appraisal of agent $j$ at time $t$.
${X}_{ij}(t) =1$ if $i$ has  positive feelings towards $j$ and  ${X}_{ij}(t)=-1$ if $i$ has negative feelings towards $j$, while  ${X}_{ij}(t)=0$ if   $i$    chooses not rely on $j$ in forming its opinion\footnote{Since we consider small-medium size networks, this formalizes the case when agent $i$ knows agent $j$ but does not find correlation between its own choices and agent $j$'s opinions, and hence chooses not to give it   any weight.}. 
We assume that for each pair of agents $(i,j)$ at each time instant $t$ the  appraisal is mutual, namely  ${X}_{ij}(t)= { X}_{ji}(t)$ $\forall i,j\in [1,N]$, and hence ${\mathbf X}(t)$  is a symmetric matrix $\forall t \geq 0$. 
  The (undirected and signed) graph ${\mathcal G}({\bf X})$, having  ${\mathbf X}$  as   adjacency matrix, represents the  {\em appraisal network} \cite{HomophilyMei}.\\
We assume that the agents express their opinions about a certain number, say $m$, of issues. This information is collected in a matrix ${\mathbf Y}(t)\in {\mathbb R}^{N\times m}$, whose $(i,j)$-th entry is the opinion that agent $i$ has about the issue $j$ at the time instant $t$. ${\mathbf Y}(t)$ is  called the \textit{opinion matrix at the time instant t} of the social network. We assume that the opinion matrix and the appraisal matrix evolve according to an intertwined dynamics expressed by the following equations
\begin{align}
{\mathbf X}(t+1) &= {\rm sgn}({\mathbf Y}(t){\mathbf Y}(t)^\top)  \label{appraisal_dy}\\
{\mathbf Y}(t+1) &= \frac{1}{N} {\mathbf X}(t+1){\mathbf Y}(t) \label{opinion_dy} 
\end{align}
  that component-wise correspond to
\begin{align}
{ X}_{ij}(t+1) &= {\rm sgn}\left(\sum_{k=1}^m{ Y}_{ik}(t){ Y}_{jk}(t)\right)  \label{appraisal_comp}\\
{Y}_{ij}(t+1) &=  \frac{1}{N} \sum_{k=1}^N{ X}_{ik}(t+1){ Y}_{kj}(t). \label{opinion_comp}
\end{align}
Equation \eqref{opinion_comp} shows that the opinion that agent $i$ has about issue $j$ at the time instant $t+1$ is a  (signed) weighted average of the opinions that all agents have about the topic $j$ at the time instant $t$,  where the weights are the appraisals that agent $i$ has about them at the time instant $t$, divided by the number of agents.

On the other hand, from equation \eqref{appraisal_comp}, we notice that the the $(i,j)$-th entry of the appraisal matrix at the time instant $t+1$, namely, the appraisal that agent $i$ has about agent $j$ at the time instant $t+1$, depends on the comparison between the opinions that agents $i$ and $j$ have about all the topics at the time instant $t$.  In particular, if the agents agree (resp. disagree) on a specific issue $k$, this will give a positive (resp. negative) contribution $Y_{ik}(t)Y_{jk}(t)>0$ (resp. $Y_{ik}(t)Y_{jk}(t)<0$), in determining the relationship between $i$ and $j$ at the time instant $t+1$. 

Essentially, this model captures the evolution of opinion-dependent time-varying graph structures. In this regard one can see analogies with the pioneering work form Hagselmann-Krause \cite{boundedconf}, in which the closeness of opinions determines the structure topology of the (unweighted) interaction  graph. On the other hand, in our model all agents potentially communicate and their opinions will rather determine the type (friendly/antagonistic) of relationship.
   Equations   \eqref{appraisal_dy} and  \eqref{opinion_dy}  can be grouped into a single equation that describes the update of the opinion matrix alone and takes the form
\begin{equation} \label{our_model}
{\mathbf Y}(t+1) = \frac{1}{N} {\rm sgn}({\mathbf Y}(t){\mathbf Y}(t)^\top) {\mathbf Y}(t).
\end{equation}
{\color{black} Equation \eqref{our_model} shows that  the mathematical abstraction of the bandwagon bias leads the intertwining between opinion dynamics and appraisal dynamics to a peculiar form of opinion dynamics model. 
}
It is immediate to notice that if ${\bf Y}(0)$ has a zero row (a situation that formalizes the case when one of the agents 
expresses no opinion   on any of the $m$ topics), then that same row remains zero in every subsequent opinion matrix ${\bf Y}(t), t\ge 0.$ 
Similarly, if ${\bf Y}(0)$ has a zero column (none of the agents expresses any judgement on a specific topic), that same column 
remains zero in all the matrices ${\bf Y}(t), t\ge 0.$  Therefore both cases are  of no interest (substantially, one can always remove the agent and/or the topic and focus on the analysis of the remaining variables) and will not be considered in the following. 
\medskip

 \begin{remark} \label{comparison} 
Compared with the model proposed and investigated in  \cite{MeiDorfler}, we have modified the law that governs the appraisal matrix update 
and how it affects the opinion dynamics in two aspects. First, we have chosen to keep into account only the signs of the mutual appraisals, rather than their absolute values. This is motivated by the fact that, in a lot of practical situations, being able to  assess the sign of the mutual appraisal is easier and more robust to modeling errors with respect to determining the numerical value associated to the tie strength. 
  Moreover, the influence  that agent $j$  can have on the opinion   agent $i$ has on a certain issue does not necessarily scale with the absolute value of ${ X}_{ij}$. Secondly, we have chosen to ``give a weight" also to the fact that a pair of agents chooses not  rely on each other's opinion, namely to the fact that ${ X}_{ij}=0$.
Since we consider small-medium size networks, this formalizes the case when agent $i$ knows agent $j$ but 
does not find correlation between its own choices and agent $j$'s opinions, and hence chooses not to give  it any weight.
In this perspective,  the fact that the mutual appraisal is $0$ is an information that should be considered and this motivates the fact that 
in the opinion dynamics update equation  \eqref{opinion_comp} each row is divided by   the overall number of agents $N$, rather than by the absolute value of its entries.
It is worth noticing that, however, since the appraisal matrix is obtained by comparing the (real valued) opinions of the agents on the various topics into play, 
and its $(i,j)$-th entry is zero only if the opinion vectors of agents $i$ and $j$ are orthogonal, a zero entry in the appraisal matrix is a very rare occurrence, as it will be confirmed by the numerical simulations at the end of the paper.\\
As we will see in the following, our   model retains all the relevant features of the model investigated in \cite{MeiDorfler}, and  it is simpler to analyse and implement.
\end{remark}

{\bf Assumption 1}(No zero rows/columns). \ {\em In the following, we will steadily assume that ${\bf Y}(0)$ is devoid of zero rows/columns. }\medskip


\begin{lemma}[No zero rows dynamics]\label{nozerorows} If ${\bf Y}(0)\in {\mathbb R}^{N\times m}$ has no zero rows, then for every $t\ge 0$ the matrix ${\bf Y}(t)$, obtained from the model 
\eqref{our_model} corresponding to the initial condition ${\bf Y}(0)$, has no zero rows.
\end{lemma}

\begin{proof}  
Suppose, by contradiction, that this is not the case, and
let $t_0 \ge 0$ be the smallest time instant such that ${\bf Y}(t_0)$ has no zero rows, but ${\bf Y}(t_0+1)$ has (at least) one zero row.
It entails no loss of generality assuming that the first row  of ${\bf Y}(t_0+1)$  becomes zero   (if not we can always resort to a relabelling of the agents
to reduce ourselves to this case). If we set ${\bf Y}:= {\bf Y}(t_0)$, this means that ${\bf Y}$ has no zero rows, but
$${\bf e}_1^\top {\rm sgn} ({\bf Y Y}^\top) {\bf Y} = {\bf 0}^\top.$$
Set ${\bf z}^\top := {\bf e}_1^\top {\rm sgn} ({\bf Y Y}^\top) \in \{-1,0,1\}^{1\times N}$. We observe that since the first row of ${\bf Y}$ is not zero then the $(1,1)$-entry of ${\bf Y Y}^\top$ is positive and hence the first entry of ${\bf z}$ is $1$. The remaining ones belong to $ \{-1,0,1\}$.
 We distinguish two cases:  either all the other entries of ${\bf z}$ are zero (Case A) or  there exist other nonzero entries in ${\bf z}$ (Case B), and in this latter case we can assume without loss of generality (if not, we can always permute the $m$ topics, namely the $m$ columns of ${\bf Y}$, to make this possible)
that 
$${\bf z}^\top = \begin{bmatrix} 1 &\vline& z_2 & \dots & z_r &\vline& 0 & \dots & 0\end{bmatrix}, \ \begin{array}{c}z_i\in \{-1,1\},\cr i\in [2,r].\end{array}$$
Condition ${\bf z}^\top {\bf Y}= {\bf 0}^\top$ implies that the columns of ${\bf Y}$ are all orthogonal to the vector ${\bf z}$.
In Case B this implies that
${\bf Y}$ can be expressed as
\be
{\bf Y} = \begin{bmatrix}
\begin{matrix} {\bf 1}_{r-1}^\top \cr
\hline
\Sigma
\end{matrix}
&\vline &\begin{matrix} {\bf 0}^\top\cr \hline {\bf 0}^\top\end{matrix} \cr
\hline
0 &\vline& I_{N-r}
\end{bmatrix} \begin{bmatrix}C_a\cr C_b\end{bmatrix}
\label{Y}
\ee
for some matrices $C_a\in {\mathbb R}^{(r-1)\times m}$ and 
 $C_b\in {\mathbb R}^{(N-r)\times m}$, where
 $\Sigma := - {\rm diag}\{   z_2, \dots,  z_r\}.$
 On the other hand, the vector ${\bf z}^\top$ and the matrix ${\bf Y}$ are related by the identity 
 ${\bf z}^\top = {\bf e}_1^\top {\rm sgn} ({\bf Y Y}^\top)=  {\rm sgn} ({\bf e}_1^\top{\bf Y Y}^\top) $, and hence it must be\\
$$\begin{bmatrix} 1 &\vline& z_2 & \dots & z_r &\vline& 0 & \dots & 0\end{bmatrix}=  {\rm sgn}\left( {\bf 1}_{r-1}^\top C_a \begin{bmatrix}C_a^\top &  C_b^\top \end{bmatrix}
 \begin{bmatrix} {\bf 1}_{r-1} &\vline& \Sigma &\vline&0\cr
 0 &\vline& 0 &\vline& I_{N-r}\end{bmatrix}\right).$$
 This implies, in particular, that 
 $$\begin{bmatrix}   z_2 & \dots & z_r  \end{bmatrix}
= {\rm sgn}( {\bf 1}_{r-1}^\top C_a C_a^\top  \Sigma),$$
or, entrywise, keeping into account the definition of $\Sigma$:
$$z_i = - {\rm sgn}( {\bf 1}_{r-1}^\top C_a C_a^\top  z_i {\bf e}_{i-1}), \quad \forall\ i\in [2,r].$$
This amounts to saying that 
$$ {\rm sgn}( {\bf 1}_{r-1}^\top C_a C_a^\top   {\bf e}_{i-1}) = -1, \quad \forall\ i\in [2,r],$$
namely 
${\bf 1}_{r-1}^\top C_a C_a^\top  \ll 0,$
by this meaning that it is a vector with all negative entries. But this would imply $\| C_a^\top {\bf 1}_{r-1}  \|^2= {\bf 1}_{r-1}^\top C_a C_a^\top {\bf 1}_{r-1} < 0$, which is clearly impossible.\\
We consider now Case A. If the only nonzero entry of ${\bf z}$ is the first one, then ${\mathbf Y}$ can be expressed as ${\mathbf Y}={\mathbf W}C_0 $, where ${\mathbf W} = [{\mathbf 0} | I_{N-1} ]^\top$ and $C_0$ is a real matrix of size $(N-1)\times m$.
By resorting to the same reasoning as in Case B, condition ${\bf z}^\top = {\bf e}_1^\top {\rm sgn} ({\bf Y Y}^\top)$ becomes
\begin{equation*}
\begin{bmatrix}
1 &0& \dots& 0
\end{bmatrix} = {\rm sgn} ({\mathbf e}_1 {\mathbf W}C_0C_{\bf 0}^\top {\mathbf W}^\top) = {\rm sgn} ({\mathbf 0}^\top),
\end{equation*}
which is impossible.
Therefore it is not possible that there exists $t_0 \ge 0$   such that ${\bf Y}(t_0)$ has no zero rows, but ${\bf Y}(t_0+1)$ has (at least) one zero row. 
\end{proof}
\bigskip

Based on the preliminary remarks and   Lemma \ref{nozerorows}, we introduce the set \cite{HomophilyMei}
$${\mathcal S}_{nz-rows} := \{{\bf Y}\in {\mathbb R}^{N\times m}: {\bf e}_i^\top {\bf Y} \ne {\bf 0}^\top, \forall\ i\in [1,N]\},$$
and in the following we will steadily assume that ${\bf Y}(0)\in {\mathcal S}_{nz-rows}$, and hence ${\bf Y}(t)\in {\mathcal S}_{nz-rows}$ for every $t\ge 0$. It is worth noticing that, differently from \cite{MeiDorfler}, we do not need to impose
that ${\bf Y}(0)\in \mathcal{Y}:= \{{\bf Y} : {\bf Y}(t) \in {\mathcal S}_{nz-rows} \forall t \geq 0\}$, since for our model it suffices to assume that ${\bf Y}(0)\in {\mathcal S}_{nz-rows}$  to guarantee that ${\bf Y}(t)\in {\mathcal S}_{nz-rows}$, $\forall t \geq 0$.
\\
Note that, as a further consequence, for every $t\ge 0$,
${\bf X}(t+1)= {\rm sgn} ({\bf Y}(t){\bf Y}(t)^\top)$  is a symmetric matrix with unitary diagonal entries,
 and hence belongs to
$\mathcal{S}_{1}^N, \, \forall t \geq 0.$ \\


\begin{remark}
{\color{black} 
The case when  there exists  $t > 0$ such that the matrix ${\mathbf Y}(t)$ has a zero column, even if ${\mathbf Y}(0)$ has no zero columns,  may arise,  but it is a   rare occurrence. This   happens if and only if one of the columns of ${\mathbf Y}(t)$ belongs to the kernel of the matrix ${\bf X}(t+1)= sgn({\mathbf Y}(t){\mathbf Y}(t)^\top)$. This means that at the time $t$ the column vector describing the opinions that the agents have on some specific topic is such that {\em for every agent $i$} the sum of the opinions of the agents trusted by $i$    equals the sum of the opinions of the agents not trusted by agent $i$.  Since the agents'opinions are arbitrary real numbers this case arises for a set of initial conditions ${\mathbf Y}(0)$ having zero measure.

 An elementary example is represented by 
 the case when
 ${\bf Y}(0) = \begin{bmatrix} 1 & \epsilon\cr 2 & -\epsilon\end{bmatrix}$, where $\epsilon$ is nonzero and sufficiently small. Correspondingly,}
we get ${\bf Y}(1) = \begin{bmatrix} 3/2 & 0\cr 3/2 & 0\end{bmatrix}$.
\end{remark}
\bigskip

After having explored these preliminary aspects regarding 
agents that become indifferent to all issues, or issues that become irrelevant to all agents,
 we want to investigate the existence   and structure of the equilibrium points for the   model \eqref{appraisal_dy}- \eqref{opinion_dy},  when starting from initial opinion matrices ${\bf Y}(0)$ satisfying Assumption 1. \medskip

\begin{definition}[Equilibrium point] A pair $({\bf Y}^*, {\bf X}^*)$ is an \emph{equilibrium point} for the   model  \eqref{appraisal_dy}- \eqref{opinion_dy}
if 
\begin{eqnarray}
{\bf X}^* &=& {\rm sgn} ({\bf Y}^* ({\bf Y}^*)^\top) \label{uno}\\
 {\bf Y}^* &=& \frac{1}{N} {\bf X}^* {\bf Y}^*. \label{due}
 \end{eqnarray}
\end{definition}
\medskip

It is interesting to notice that the only possible nontrivial equilibrium points for the model are those
  that correspond to a structurally balanced configuration of the appraisal network ${\mathcal G}({\bf X}^*)$. Moreover, the appraisal network is  necessarily complete, namely each agent needs to express its appraisal   towards all the other agents.
  \medskip
  
\begin{proposition}[Equilibrium equivalence conditions] \label{equilibria}
A pair $({\bf Y}^*, {\bf X}^*)\ne {\bf (0,0)}$  is an equilibrium point for the   model  \eqref{appraisal_dy}- \eqref{opinion_dy} if and only if
\begin{itemize}
\item[i)] ${\bf X}^*={\bf p}{\bf p}^\top$, for some ${\bf p}\in \{-1,1\}^N$;
\item[ii)] ${\bf Y}^*={\bf p}\begin{bmatrix} a_1 & a_2 & \dots & a_m\end{bmatrix}$, for some $a_i\in {\mathbb R}, \sum_{i=1}^m a_i^2\ne 0$.
\end{itemize}
 \end{proposition}
 
 \begin{proof} It is immediate to observe that if i) and ii) hold, then the identities \eqref{uno} and \eqref{due} hold.
 \\
 Conversely, assume that the pair $({\bf Y}^*, {\bf X}^*)$ is an equilibrium point. Then \eqref{due} holds, but this means that the nonzero columns of ${\bf Y}^*$ are eigenvectors of $\frac{1}{N} {\bf X}^*$ corresponding to the unitary eigenvalue.
 This means that $1\in \sigma\left(\frac{1}{N} {\bf X}^*\right)$ and therefore, by Lemma \ref{eig1} in the Appendix, i) holds.
 On the other hand, by replacing  the matrix ${\bf X}^*$ in \eqref{due} with ${\bf p}{\bf p}^\top$, we obtain
 ii).\end{proof}
 \bigskip

 {\color{black}
\begin{remark} The non-trivial equilibrium points of the model are \emph{modulus consensus} configurations, see, e.g., \cite{MeiDorfler}. This is also what happens for  equilibrium points   in \cite{Altafini2013} and \cite{MeiDorfler}. Moreover, when the model converges to the non-trivial equlibrium configurations, the sign distribution of the opinions mirrors the network partition into factions. \\
As in Altafini's model \cite{Altafini2013}, and as it will be clear in the following, the system dynamics either achieves modulus consensus (in a finite number of steps) or  converges to zero (asymptotically). \end{remark}
}

 \begin{remark} 
  This situation is different from the one that arises with the model investigated in \cite{Bullo2020},     \cite{MeiDorfler} and \cite{not_all_to_all}.
In \cite{MeiDorfler} (see Remark 4 and section IV in \cite{MeiDorfler}) 
 the equilibrium points identified in the previous Proposition \ref{equilibria} are not the only possible ones. Indeed, for the model explored in
\cite{MeiDorfler} the case may occur that the matrix ${\bf X}^*$ at the equilibrium corresponds to a non connected graph, whose connected components however achieve structural balance.
 As we will see  later (see Remark \ref{disconnected}), if ${\mathcal G}({\bf X}^*)$ becomes disconnected then both the opinion matrix and the appraisal matrix converge to zero. 
 On the other hand,   for   the binary model in \cite{not_all_to_all}, convergence to a non-all-to-all structurally balanced network is also possible while    the topological structure of the associated   graph is time invariant.  Also, under some conditions, convergence to ``two-factions" structural balance  is obtained in finite time.
The signed formation game in \cite{Bullo2020} dynamically drives the network towards clustering balance.
 \end{remark}
 \medskip

We want now to show that the model we have proposed cannot exhibit periodic solutions and hence limit cycles.
To prove this result we need a preliminary lemma, that will be useful also for   the subsequent analysis.
 \bigskip
 
\begin{lemma}[Upper bounded opinion dynamics]\label{lemmaAB}For every $j\in[1,m]$ and every $t\ge 0$
\\
i)\ 
 \be
 \max_{i\in [1,N]} |{ Y}_{ij}(t+1)| \le  \max_{i\in [1,N]} |{Y}_{ij}(t)|.
 \label{upper}
 \ee
ii)\ Condition
 $$ \max_{i\in [1,N]} |{Y}_{ij}(t+1)| =  \max_{i\in [1,N]} |{ Y}_{ij}(t)|   \ne 0$$
 holds if and only if\\
 (a) ${\bf Y}(t) {\bf e}_j= {\bf p} \cdot \mu_j, \ \exists\ {\bf p}\in \{-1,1\}^N$ and $\mu_j > 0;$ and\\
 (b) once we set $h := {\rm argmax}_{i\in [1,N]} |{Y}_{ij}(t+1)|$ then ${\bf e}_h^\top {\bf X}(t+1)$ has no zero entries 
 and
 ${\bf e}_h^\top {\bf X}(t+1)= p_h\cdot {\bf p}^\top$.
  \end{lemma}

 \begin{proof}
 i)\ From equation 
 \eqref{opinion_dy} it follows that
 \begin{eqnarray*}
|{ Y}_{ij}(t+1)| &=& \lvert \frac{1}{N} \sum_{k=1}^N { X}_{ik}(t+1){ Y}_{kj}(t) \lvert    \leq \frac{1}{N}  \sum_{k=1}^N |{ X}_{ik}(t+1)||{ Y}_{kj}(t) | \\
&\leq& \frac{1}{N}  \sum_{k=1}^N |{ Y}_{kj}(t) |  \leq \frac{1}{N} N\max_{k }|{ Y}_{kj}(t) | = \max_{k}|{ Y}_{kj}(t) |,
\end{eqnarray*}
and hence \eqref{upper} holds.\\
ii) Set $h := {\rm argmax}_{i\in [1,N]} |{ Y}_{ij}(t+1)|$. Then 
 $|{ Y}_{hj}(t+1)| = {\rm max}_{i\in [1,N]} |{\ Y}_{ij}(t+1)|$ coincides with   $\max_{i\in [1,N]} |{ Y}_{ij}(t)|$ if and only if
$$\sum_{\ell=1}^N |{ X}_{h\ell}(t+1)||{ Y}_{\ell j}(t) | = N \cdot  \max_{i\in [1,N]} |{Y}_{ij}(t)|$$
and this is possible if and only if 
 all the entries in the $j$-th column of ${\bf Y}(t)$ have the same absolute value $\mu_j> 0$ (and this leads to (a), for some suitable vector ${\bf p}$) 
and all the terms 
${ X}_{h\ell}(t+1){ Y}_{\ell j}(t)$, $\ell\in [1,N]$, have the same sign. But this latter condition means that 
  ${\bf e}_h^\top {\bf X}(t+1)$ either coincides with   ${\bf p}^\top$ or with its opposite, and since ${ X}_{hh}(t+1)=1$ this means that condition (b) holds.
 \end{proof}
\bigskip

We are now in a position to prove the following result.
\bigskip

\begin{proposition}[Aperiodicity in opinion dynamics] Suppose that there exist $\bar t\ge 0$, $T\ge 1$ and nonzero matrices $\tilde{\bf Y}_i\in {\mathbb R}^{N\times m}, i\in [1,T],$ such that
$${\bf Y}(\bar t +i) = \tilde{\bf Y}_i, i\in [1,T], \quad {\rm and}\quad
{\bf Y}(\bar t+T+1)= \tilde{\bf Y}_1,$$
namely from $\bar t +1$ onward the sequence of matrices $\{{\bf Y}(t)\}_{t\ge \bar t+1}$ becomes periodic of period $T$, then $T=1$, 
namely the sequence becomes constant.
\end{proposition}
 
 \begin{proof}
 From Lemma \ref{lemmaAB}, part i), we can claim that for every $j\in [1,m]$ and every $t\ge 0$
$$\begin{array}{l}
 \max_{\ell\in [1,N]} |{ Y}_{\ell j}(t+T+1 )| \le   \max_{\ell \in [1,N]} |{ Y}_{\ell j}(t+T)| \cr
 \cr
 \le ... \le \max_{\ell \in [1,N]} |{ Y}_{\ell j}(t+2)| \le  \max_{\ell \in [1,N]} |{ Y}_{\ell j}(t+1)|.
 \end{array}$$
But since  for $t=\bar t$ we have ${\bf Y}(\bar t+T+1 )= {\bf Y}(\bar t+1 )=\tilde {\bf Y}_1$ and hence the two extremes in the previous sequence of inequalities coincide, it follows that all the symbols $\le$ are equalities, namely
$$ \max_{\ell\in [1,N]} |[\tilde{\bf Y}_i]_{\ell j}| =\mu_j > 0, \qquad \forall\ j\in [1,m], \ \forall\ i\in [1,T].$$
This also implies, see Lemma \ref{lemmaAB} part ii), that,  for every non zero column in $\tilde{\bf Y}_i$,
\be
\tilde{\bf Y}_i{\bf e}_j = {\bf p}_i \cdot \mu_j, \qquad \exists\ {\bf p}_i\in \{-1,1\}^N,  \mu_j >0,
\label{one}
\ee
and that, for every $h\in [1,N]$, one has ${\bf e}_h^\top  {\rm sgn}(\tilde{\bf Y}_i\tilde{\bf Y}_i^\top) = [{\bf p}_i]_h \cdot {\bf p}_i^\top$.
This implies that for every $i\in [1,T]$ 
$${\rm sgn}([\tilde{\bf Y}_i\tilde{\bf Y}_i^\top]) = {\bf p}_i {\bf p}_i^\top, \qquad \exists\ {\bf p}_i\in \{-1,1\}^{N\times N}.$$
Consequently\footnote{The expression $i+1 \mod T$ means the remainder of  $i+1$ when divided by $T$.}
\be
\tilde{\bf Y}_{(i+1 \mod T)} = \frac{1}{N} {\rm sgn}([\tilde{\bf Y}_i\tilde{\bf Y}_i^\top]) \tilde{\bf Y}_i= \frac{1}{N}  {\bf p}_i {\bf p}_i^\top \tilde{\bf Y}_i.
\label{two}
\ee
So, by comparing \eqref{one} and \eqref{two}  one gets 
that every matrix $\tilde{\bf Y}_i, i\in [1,T],$ takes 
the form 
$$\tilde{\bf Y}_i= {\bf p}_i \begin{bmatrix} a_1^{(i)} & \dots & a_m^{(i)}\end{bmatrix}, \qquad \exists\ {\bf p}_i\in \{-1,1\}^N, a_k^{(i)} \in {\mathbb R}, $$
but this also implies that
$$
\begin{array}{l}
\tilde{\bf Y}_{(i+1 \mod T)} = \frac{1}{N} {\rm sgn}([\tilde{\bf Y}_i\tilde{\bf Y}_i^\top]) \tilde{\bf Y}_i
=  \frac{1}{N} {\rm sgn}({\bf p}_i {\bf p}_i^\top \cdot \sum_k [a_k^{(i)}]^2) {\bf p}_i \begin{bmatrix} a_1^{(i)} & \dots & a_m^{(i)}\end{bmatrix} \cr
\cr
= 
\frac{1}{N}  {\bf p}_i {\bf p}_i^\top  {\bf p}_i \begin{bmatrix} a_1^{(i)} & \dots & a_m^{(i)}\end{bmatrix}
= {\bf p}_i \begin{bmatrix} a_1^{(i)} & \dots & a_m^{(i)}\end{bmatrix}
= \tilde{\bf Y}_i.
\end{array}$$

So, all matrices $\tilde{\bf Y}_i$ coincide.
 \end{proof}
\medskip

\section{Convergence to an equilibrium in a finite number of steps}\label{convergence_FT}

  We want to explore   under what conditions the equilibrium can be reached in a finite number of steps.
It is easy to see that
 if there exists a time instant $t_0 \geq 0$ such that 
${\mathbf Y}(t_0+1) = {\bf Y}(t_0)\ne 0$ then ${\bf Y}(t)= {\bf Y}(t_0)=: {\bf Y}^*$ for every $t\ge t_0$. Consequently, 
also ${\bf X}(t)$ becomes constant starting at $t=t_0+1$, and it coincides with ${\mathbf X}^*:= {\rm sgn}({\bf Y}(t_0){\bf Y}(t_0)^\top).$
\\
However, the converse is not true:  if  the appraisal matrix   becomes constant at some time $t_0\ge 0$, the opinion matrix ${\bf Y}(t)$ can still keep evolving for $t\ge t_0$. This situation is illustrated in Example \ref{example1}, below.

 As a matter of fact,
  if there exists a time instant $t_0 \geq 0$ such that ${\mathbf X}(t) = {\mathbf X}^*, \forall t \geq t_0$, we can only claim that   ${\mathbf Y}(t+1) =  \frac{1}{N} {\mathbf X}^*{\mathbf Y}(t)$. Equivalently, if we denote by ${\mathbf y}_j(t)$, the $j$-th column of the matrix ${\mathbf Y}(t)$, then the dynamics expressed by equation \eqref{appraisal_dy} decomposes into $m$ linear time invariant systems of the form
\begin{equation}\label{componet_yse}
{\mathbf y}_j(t+1) = \frac{1}{N} {\mathbf X}^*{\mathbf y}_j(t), \, \forall \ j \in [1,m].
\end{equation} 
As ${\mathbf X}^*\in S_1^N$, the matrix $\frac{{\mathbf X}^*}{N}$  is symmetric and hence diagonalizable. Moreover,
by Gershgorin Circle theorem, all its (real) eigenvalues $\lambda_i, i\in [1,N],$ satisfy
\begin{equation*}
\lvert \lambda_i-\frac{1}{N}\lvert \leq \frac{N-1}{N} \iff -\frac{N-2}{N}\leq \lambda_i\leq 1, \, \forall i \in [1,N].
\end{equation*}
As a consequence, two cases may arise.  The first case is  the one depicted in Example \ref{example1}, namely the case when  the systems in \eqref{componet_yse} are  asymptotically stable, which means that ${\bf Y}(t)$ asymptotically converges to $0$ (and hence $\lim_{t\to +\infty} {\bf X}(t)=0 \ne {\bf X}^*$). 
\medskip

\begin{example}\label{example1}
 Let us consider   the case $N = m = 3$, with 
\begin{equation*}
\mathbf{Y}(0) = \begin{bmatrix} 1.41 & -1.21 & 0.49\\
1.42 & 0.72 & 1.03 \\
0.67 & 1.63 & 0.73
\end{bmatrix}.
\end{equation*} 
It turns out that $\forall t \geq 1$ 
\begin{equation*}
\mathbf{X}(t) =\mathbf{X}^* =\begin{bmatrix} 1 & 1 & -1\\
1 & 1& 1 \\
-1 & 1& 1
\end{bmatrix}
\end{equation*} and $\sigma(1/3\cdot {\bf X}^*)= (-1/3, 2/3,2/3)$, and indeed   for $t \geq 14$ we have $Y_{ij}(t) = o(10^{-2}), \forall i,j \in [1,3]$.  
\end{example}\bigskip

The second possible situation is when
$\frac{{\mathbf X}^*}{N}$ is simply (but not asymptotically) stable.
 This amounts to saying that $1$ is a (simple) eigenvalue of $\frac{{\mathbf X}^*}{N}$, and hence by Lemma \ref{eig1}, ${\mathbf X}^*$ takes the form 
 ${\mathbf X}^* =  {\mathbf p} {\mathbf p}^\top$, $\exists\ {\mathbf p} \in \{-1,1\}^N$. In this case, the convergence is not asymptotic but instantaneous.
 In fact, it is sufficient that $\frac{1}{N} {\bf X}(t_0)$ becomes simply (but not asymptotically) stable at a single time instant, to ensure the instantaneous convergence of ${\bf Y}(t)$ to an equilibrium condition.
 \medskip

 \begin{proposition}[Equilibrium points characterization]\label{Xrango1} If there exists
${\color{black}t_0 > 0}$ such that $\frac{1}{N} {\bf X}(t_0)$,  with $ {\bf X}(t_0)\in {\mathcal S_{1}^N}$, is simply (but not asymptotically) stable,    then
$({\bf X}^*,{\bf Y}^*) := ({\bf X}(t_0),\frac{1}{N} {\bf X}(t_0) {\bf Y}(t_0-1))$ is an equilibrium point.
%
 \end{proposition}

\begin{proof}  By Lemma \ref{eig1} in the Appendix, we know that 
if $\frac{1}{N} {\bf X}(t_0)\in {\mathcal S_{1}^N}$ is simply stable or, equivalently, $1\in \sigma(\frac{1}{N} {\bf X}(t_0))$, then there exists a vector $ {\mathbf p} \in \{-1,1\}^{N}$ such that ${\mathbf X}(t_0)=  {\mathbf p} {\mathbf p}^\top$. 
On the other hand, if ${\bf X}(t_0)=   {\mathbf p} {\mathbf p}^\top$, then
$${\bf Y}(t_0)= \frac{1}{N} {\mathbf p} {\mathbf p}^\top {\bf Y}(t_0-1) = {\mathbf p} [a_1, \dots, a_m],$$
where
$$ [a_1, \dots, a_m] := \frac{1}{N} {\mathbf p}^\top  {\mathbf Y}(t_0-1).$$
Therefore 
$({\bf X}^*,{\bf Y}^*) := ({\bf X}(t_0),\frac{1}{N} {\bf X}(t_0) {\bf Y}(t_0-1))$ is an equilibrium point.
\end{proof}
\medskip

\begin{remark}
 If $m=1$ the model reaches the equilibrium in one step.   When so, in fact ${\bf X}(1) = {\rm sgn}({\bf Y}(0){\bf Y}^\top(0))  = {\bf p}{\bf p}^\top$, where ${\bf p} :=  {\rm sgn}({\bf Y}(0))$. 
  Numerical    simulations at the end of the paper will show   that, when $m>1$, namely multiple topics are considered, convergence to structural balance is almost surely guaranteed, and it occurs in a rather   small number of steps even for   medium size networks (e.g. $N=100$).
\end{remark}

  \begin{remark} \label{disconnected}        Gershgorin Circle theorem also allows to say that if $\mathbf{X}^*$ is the adjacency matrix of a disconnected graph, all the eigenvalues   of the matrix $\frac{{\mathbf X}^*}{N}$    lie in the circle of the complex plane
of center the origin and radius $\frac{N-1}{N}$ (or smaller),  and hence   $\frac{{\mathbf X}^*}{N}$  is necessarily an asymptotically stable matrix.  
\end{remark}

 Theorem \ref{finite_time_balance} summarizes the main results of this section.  

  \begin{theorem}[Main theorem]\label{finite_time_balance}
The following conditions are equivalent
\begin{itemize}
\item[i)] there exists a time instant $t_0\ge 0$ such that $1\in \sigma (\frac{1}{N}{\mathbf X}(t_0))$; 
\item[ii)] there exists a time instant $t_0\ge 0$ such that ${\mathbf Y}(t_0) = {\mathbf Y}(t_0+1)$;
\item[iii)] the opinion-appraisal dynamic model  \eqref{appraisal_dy}- \eqref{opinion_dy} converges in finite time to an equilibrium $({\bf X}^*, {\bf Y}^*)$; 
\item[iv)] the opinion-appraisal dynamic model \ \eqref{appraisal_dy}- \eqref{opinion_dy} converges in finite time to an equilibrium $({\bf X}^*, {\bf Y}^*)$, with
${\mathbf X}^* =  {\mathbf p} {\mathbf p}^\top$ and  ${\mathbf Y}^* = {\mathbf p} [a_1, \dots, a_m]$,
$\exists\ {\mathbf p} \in \{-1,1\}^N$, and  $a_i\in {\mathbb R}, i  \in [1,m]$, with $\sum_{i=1}^m a_i^2\ne 0$. \end{itemize}
\end{theorem}

\begin{proof} iv) $\Leftrightarrow$ iii) follows from Proposition \ref{equilibria}.
iii) $\Rightarrow$ ii) is obvious, while the converse has been  commented upon at the beginning of the section.\\
i) $\Rightarrow$ iv) follows from Proposition \ref{Xrango1}, while 
iv) $\Rightarrow$ i) is obvious.
\end{proof}
 \medskip

\section{Long term behavior}\label{convergence_inf}

  
  In the previous section, we have investigated what happens if either ${\bf Y}(t)$ 
  or ${\bf X}(t)$ 
  become constant starting at some time instant. In the former case the overall system  \eqref{appraisal_dy}- \eqref{opinion_dy}  reaches the equilibrium in a finite number of steps.
  In the latter  case a nontrivial equilibrium is reached if and only if ${\bf X}(t)$ at some point  becomes structurally balanced. Differently the opinion matrix asymptotically converges to zero.
We want to investigate now if a nontrivial equilibrium can be reached asymptotically, but not in a finite number of steps. \\
 An immediate consequence of the analysis of the previous section  is that if the sequence of appraisal matrices 
$\{{\bf X}(t)\}_{t\ge 1}$ does not converge in a finite number of steps then  $\frac{{\bf X}(t)}{N}$ is an asymptotically stable matrix for every $t\ge 1$. This means that
if we define the set
\be {\mathcal S}_{stable} := {\mathcal S}_1^N \setminus \{ {\mathbf X} \in {\mathcal S}_1^N: {\mathbf X} = {\mathbf p}{\mathbf p}^\top, \exists \, {\mathbf p} \in \{ -1,1 \}^N \},
\label{Sstable}
\ee
then ${\bf X}(t)\in {\mathcal S}_{stable}$   for every $t\ge 1$.
 \bigskip

\begin{proposition}[Zero vanishing condition]\label{lyapunov}
If for every $t \geq 0,$ ${\mathbf X}(t) \in {\mathcal S}_{stable}$,
 then $\lim_{t \rightarrow + \infty} {\mathbf Y}(t) = 0.$
\end{proposition}

\begin{proof}
For every $j
\in [1,m]$, let us define $\mu_j(t):= \max_{i\in[1,N]} |{Y}_{ij}(t)|$, and let us introduce the (generalized) Lyapunov function for the system in equation \eqref{our_model}, $V : {\mathbb R}^{N \times m} \rightarrow {\mathbb R} $, defined as
$$V({\mathbf Y}(t)) := \sum_{j=1}^m \mu_j(t).$$ 
We notice that $V({\mathbf Y}) \geq 0$, $\forall {\mathbf Y} \in {\mathbb R}^{N \times m}$ and that $V({\mathbf Y})=0$ if and only if ${\mathbf Y}= 0$. 
Define $\Delta_2  V({\mathbf Y}(t)) :=V({\mathbf Y}(t+2))- V({\mathbf Y}(t))$. We want to prove that $\Delta_2 V({\mathbf Y}(t)) <0$, $\forall t \geq 0$. 

By Lemma \ref{lemmaAB} it immediately follows   that $\Delta_2 V({\mathbf Y}(t)) = \sum_{j=1}^m \mu_j(t+2)- \mu_j(t) \leq 0$. 
We show now that there is not a time instant $t_0 \geq 0$ such that $\Delta_2 V({\mathbf Y}(t))=0$. If this were the case, in fact, this would mean that  $\forall j \in [1,m]$, $\mu_j (t_0+2) = \mu_j(t_0)$ and therefore $\mu_j (t_0+2)= \mu_j (t_0+1) = \mu_j (t_0) =: \mu_j$. As a consequence of Lemma \ref{lemmaAB} we deduce that
\begin{itemize}
\item[a)] ${\mathbf Y}(t_0){\mathbf e}_j = {\mathbf p}_j(t_0)\mu_j$, $\exists\ {\mathbf p}_j (t_0) \in \{-1,1\}^N$,\\
\\
$\,\,\,\,\,\,\,\,\,\,\,\,\,\,{\mathbf Y}(t_0+1){\mathbf e}_j = {\mathbf p}_j(t_0+1)\mu_j$, $\exists\ {\mathbf p}_j (t_0+1) \in \{-1,1\}^N,$
\\
\item[b)] $\forall h \in [1,N]$, ${\mathbf e}_h^T {\mathbf X}(t_0+1) = p_h {\mathbf p}_j(t_0)^\top,$
\end{itemize}	
from which it follows that ${\mathbf X}(t_0+1) = {\mathbf p}_j(t_0) {\mathbf p}_j(t_0)^\top$ for every $j\in [1,m]$.
 But then ${\mathbf X}(t_0+1) = {\mathbf p}(t_0) {\mathbf p}(t_0)^\top$ with ${\mathbf p}(t_0) = \pm {\mathbf p}_j(t_0)$, $\forall j \in [1,m]$, that means that ${\mathbf X}(t_0+1)$ is structurally balanced and hence it does not belong to ${\mathcal S}_{stable}$, thus contradicting the hypotheses.
Consequently, it must be $\Delta_2 V ({\mathbf Y}(t))<0$, $\forall t \geq 0$. Finally, by defining $\Delta_1 V ({\mathbf Y}(t)) := V({\mathbf Y}(t+1))-V({\mathbf Y}(t)) $ we get that
$$\Delta_2 V ({\mathbf Y}(t))+\Delta_1 V ({\mathbf Y}(t)) < 0, \,\,\, \forall t \geq 0,$$
so the thesis follows as a direct consequence of Theorem 2.1 in \cite{Parrilo2008}.
\end{proof}
\bigskip

 Summarizing, Theorem \ref{finite_time_balance} and Proposition \ref{lyapunov} show that either there exists a time instant $t_0$ such that $\forall \, t \geq t_0$, ${\mathbf X}(t) = {\mathbf p}{\mathbf p}^\top$ and consequently ${\mathbf Y}(t) ={\mathbf p}[a_1, a_2, \dots, a_m] $, $a_i \in \mathbb{R}$, $\sum_i a_i^2 \neq 0$, otherwise, if a time instant $t$ such that ${\mathbf X}(t)$ reaches the structural balance does not exist, then ${\mathbf Y}(t)$ converges to zero as time goes to infinity.

\section{Simulations}
In this section we show the outcome of Monte Carlo simulations in order to validate the convergence properties of the model. Figure \ref{simulations} shows how the average  number of iterations needed  in order to reach a structural balanced configuration   over the total number of $30000$ simulations  varies as  a function of the number of topics $m \in [1,10]$, for networks involving $N = 9, 20, 100$ agents. Simulations are based on initial conditions ${\bf Y}(0)$ with entries independently drawn from a Gaussian random variable with zero mean and standard deviation $\sigma  =10$, namely $Y_{ij}(0)\sim \mathcal{N}(0,100)$. It turns out that, in accordance with the Chernoff bound, by running $30000$ simulations, the estimated probability $\hat{p}$ to reach a structurally balanced configuration is equal to $1$ with accuracy $\epsilon =0.01$ and confidence level $1-\delta = 0.99$,  namely $P(\lvert \hat{p}-p\lvert \leq \epsilon)\geq 1-\delta$, for the case of $N=20,100$ agents, regardless of the number of topics taken into account while $\hat{p}$ is  greater than or equal to $0.98$   for all $m \in[1,10]$, for $N=9$, with the same accuracy and confidence interval. 
\begin{center}
\begin{figure}[h!]
\center
\includegraphics[scale=0.25]{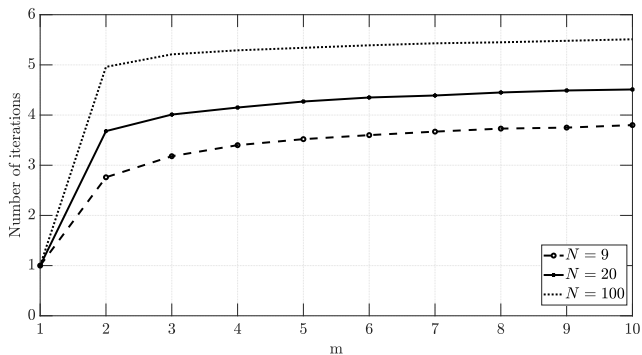}
\caption{Average number of iterations, over 30000 simulations, needed in order to reach a structural balance configurations for the cases $N = 9,20,100$ and $m \in[1,10]$.}\label{simulations}
\end{figure}
\end{center}

\section{Conclusion} \label{conclusion}

In this paper we have proposed a modified version of Liu et al. model \cite{MeiDorfler}
 for the interplay between homophily-based appraisal dynamics and influence-based opinion dynamics.
 In order to update the agents' opinions on a numbers of issues, only the signs (and not the values) of  the agents' mutual appraisals are used. This simplified model
 retains all the main characteristics of the original model, is simpler to analyse and implement, leads to the same kind of  nontrivial structurally balanced equilibria 
 as   in  \cite{MeiDorfler}, but rules out  nontrivial equilibria that correspond to disconnected socially balanced networks.
 Furthermore, nontrivial equlibria can always be reached in a finite number of steps,  while the case when all opinions and appraisals converge to zero corresponds to sets of initial conditions of zero measure.

\section*{Appendix}

\begin{lemma}[Rank-one matrices with special structures]\label{eig1}
Given a matrix ${\bf M}\in S_1^N$, if $1\in \sigma \left(
\frac{1}{N} {\bf M}\right)$, then
${\bf M} = {\bf p}{\bf p}^\top$ for some ${\bf p}\in \{-1,1\}^N$, and hence ${\bf M}$ has no zero entries and 
$\sigma({\bf M}) = (0, \dots, 0,1)$.
\end{lemma}

\begin{proof}
Let ${\bf v} := \begin{bmatrix} v_1 & v_2 & \dots & v_N\end{bmatrix}^\top\in {\mathbb R}^N, {\bf v}\ne 0,$ be an eigenvector of $\frac{1}{N}{\mathbf M}$ corresponding to the unitary eigenvalue, or equivalently of ${\bf M}$ corresponding to $N$. Then
${\bf M}  {\bf v} = N  {\bf v}.$
Let $h := {\rm argmax}_{i\in [1,N]} |v_i|.$ Then
condition
$$N v_h = \sum_{i=1}^N { M}_{hi} v_i = v_h + \sum_{i=1\atop i\ne h}^N { M}_{hi} v_i$$
  holds if and only if (a) $|v_i|=|v_h|$ for every $i\in [1,N]$;
  (b) ${ M}_{hi}  \ne 0$ for every $i\in [1,N]$, and ${\rm sgn}({M}_{hi}){\rm sgn}(v_i) = {\rm sgn}(v_h)$. \\
  This implies that
  ${\bf v} = {\bf p} m$ for some ${\bf p}\in \{-1,1\}^N$ and some $m>0$ and ${\bf e}_h^\top {\bf M} = {\rm sgn}(v_h) {\bf p}^\top 
  =    p_h {\bf p}^\top$.
  \\
  On the other hand, since condition (a) holds, this means that every index $j\in [1,N]$ is ${\rm argmax}_{i\in [1,N]} |v_i|,$
  and hence all the rows of ${\bf M}$ satisfy  ${\bf e}_i^\top {\bf M} =  p_i {\bf p}^\top$. This implies that
  ${\bf M} = {\bf p}{\bf p}^\top$, and the rest immediately follows.
\end{proof}
\bigskip

\bibliographystyle{plain} 

 \bibliography{Refer168}

\end{document}